\documentclass[acmsmall]{acmart}

\setcopyright{cc}
\setcctype{by}
\acmJournal{PACMMOD}
\acmYear{2025} \acmVolume{3} \acmNumber{2 (PODS)} \acmArticle{93} \acmMonth{5} \acmPrice{}\acmDOI{10.1145/3725230}

\usepackage{lineno}
\usepackage{balance}
\usepackage{etoolbox}
\usepackage{romannum}
\usepackage{bbm}
\usepackage{multirow}

\usepackage{blindtext}
\usepackage{mathtools}
\usepackage{enumitem}
\usepackage{bigstrut}
\usepackage{arydshln}
\usepackage{color}
\usepackage{graphicx}
\usepackage{tikz-cd}
\usepackage{algpseudocode}
\usetikzlibrary{arrows.meta}
\usetikzlibrary{calc}
\usetikzlibrary{decorations.markings}
\usetikzlibrary{math}
\usepackage{subcaption}
\usetikzlibrary{positioning,calc}
\usetikzlibrary{decorations.text}
\usetikzlibrary{decorations.pathmorphing,decorations.pathreplacing}
\usetikzlibrary{arrows,petri, topaths,fit}

\usetikzlibrary{automata, positioning,arrows,shapes,decorations.pathmorphing}

\definecolor{mycolor}{rgb}{0.122, 0.435, 0.698}
\makeatletter
\newcommand{\mybox}[1]{%
  \setbox0=\hbox{#1}%
  \setlength{\@tempdima}{\dimexpr\wd0+13pt}%
  \begin{tcolorbox}[colframe=mycolor,boxrule=0.5pt,arc=4pt,
      left=6pt,right=6pt,top=3pt,bottom=3pt,boxsep=0pt,width=\@tempdima]
    #1
  \end{tcolorbox}
}

\newcommand{\introparagraph}[1]{\vspace{0.7mm} \noindent \textbf{\em #1.}}

\newcommand{\TRUE}[0]{\mathsf{TRUE}}

\newcommand{\bC}{\mathbb{C}}
\newcommand{\bB}{\mathbb{B}}
\newcommand{\bS}{\mathbb{S}}
\newcommand{\bT}{\mathbb{T}}

\newcommand{\bx}{\mathbf{x}}

\newcommand{\RPQ}{\mathsf{RPQ}}
\newcommand{\DFA}{\mathsf{DFA}}

\newcommand{\CFG}{\mathsf{CFG}}
\newcommand{\TC}{\mathsf{TC}}

\newcommand{\id}[1]{\langle{#1}\rangle}

\newcommand{\obtainedfrom}{\text{ :- } }

\newcommand{\CQ}{\mathsf{CQ}}
\newcommand{\UCQ}{\mathsf{UCQ}}

\begin{document}
\pagenumbering{arabic}

\received{December 2024}
\received[revised]{February 2025}
\received[accepted]{March 2025}


\title{Circuits and Formulas for Datalog over Semirings}

\author{Austen Z. Fan}
\affiliation{%
  \institution{University of Wisconsin--Madison}
  \city{Madison}
  \country{USA}
  }
\email{afan@cs.wisc.edu}

\author{Paraschos Koutris}
  \affiliation{%
    \institution{University of Wisconsin--Madison}
    \city{Madison}
    \country{USA}
    }
  \email{paris@cs.wisc.edu}

\author{Sudeepa Roy}
\affiliation{%
    \institution{Duke University}
    \city{Durham}
    \country{USA}
}
\email{sudeep@cs.duke.edu}



\begin{abstract}
In this paper, we study circuits and formulas for provenance polynomials of Datalog programs. We ask the following question: given an absorptive semiring and a fact of a Datalog program, what is the optimal depth and size of a circuit/formula that computes its provenance polynomial? We focus on absorptive semirings as these guarantee the existence of a polynomial-size circuit. Our main result is a dichotomy for several classes of Datalog programs on whether they admit a formula of polynomial size or not. We achieve this result by showing that for these Datalog programs the optimal circuit depth is either $\Theta(\log m)$ or $\Theta(\log^2 m)$, where $m$ is the input size. We also show that for Datalog programs with the polynomial fringe property, we can always construct low-depth circuits of size $O(\log^2 m)$. Finally, we give characterizations of when Datalog programs are bounded over more general semirings.
\end{abstract}

\maketitle
\allowdisplaybreaks

\section{Introduction}
\label{sec:intro}

In this work, we study how to efficiently store provenance information for Datalog programs when interpreted over semirings. Computing provenance information for query results is a fundamental task in data management, with applications to various domains~\cite{GreenKT07}. In the context of Datalog, the presence of recursion in the queries causes problems in how efficiently we can capture provenance. This is because provenance needs to keep track of all the possible ways in which we can derive a given fact, which for Datalog programs can be exponentially large or even unbounded with respect to the input size. Contrast this with the task of storing provenance information for Conjunctive Queries ($\CQ$s) or even Unions of Conjunctive Queries ($\UCQ$s), where there are only polynomially many such derivations {considering data complexity \cite{Vardi88}}. 

In more technical terms, the provenance of a fact in Datalog is captured by a provenance polynomial $p$ that is a sum of products, the sum being over all possible proof trees (or derivation trees) of the fact, and the product being over all possible leaves of the proof tree, which always correspond to input facts {(see Section~\ref{sec:provenance_polynomial} for the exact definition.)}. For example, consider the Datalog program that computes the transitive closure $T(x,y)$ of a directed graph: the provenance polynomial for a fact $T(s,t)$ encodes all possible paths from $s$ to $t$. In general, the number of proof trees can be infinite. In the transitive closure example, the presence of cycles in the graph can lead to {paths with unbounded length}. By interpreting the polynomial $p$ over different semirings, we obtain different information from $p$. However, there exist semirings for which the value of the polynomial $p$ is not well-defined (for example, the arithmetic semiring). To overcome this issue, we focus on semirings where the infinite sum always takes a finite value in the semiring; in this paper, we will work with {\em absorptive} (or {\em $0$-stable}~\cite{KhamisNPSW24}) semirings\footnote{{Examples of semirings beyond absorptive semirings where the provenance has a bounded representation is given in prior work \cite{DeutchMRT14} without a generalized property, which remains an interesting future work.}}.

To efficiently store these polynomials, we turn into circuits and formulas. These can be viewed as compressed data structures that represent the polynomial, with the guarantee that the polynomial value can be computed in time linear to the representation size.
In a seminal work~\cite{DeutchMRT14}, it was shown that all provenance polynomials over absorptive semirings can be represented by circuits of polynomial size. This means that information can be compressed very efficiently, since the polynomials written as a sum-of-products can be exponentially large. On the other hand, the provenance polynomial for transitive closure requires super-polynomial size formulas~\cite{KarchmerW90}. This motivates the following question:  
\begin{quote} 
\centering
   {\em Which Datalog programs admit polynomial-size formulas and which not?}
\end{quote}

An ideal answer would classify every Datalog program in one of two categories, and prove the dichotomy for various semirings -- not only for the Boolean semiring. 
This paper answers this question for several fragments of Datalog. A key idea underlying our approach is that to bound the size of a formula we need to be able to bound the depth of a circuit (the circuits we consider have fan-in two). In particular, if $m$ is the input size, a circuit of depth $O(\log m)$ implies a polynomial-size formula, while a circuit lower bound of $\Omega(\log^{1+\epsilon} m)$ for some $\epsilon >0$ implies a super-polynomial lower bound. Hence, we can rephrase our initial question  as: {\em which Datalog programs admit circuits of depth $O(\log m)$, and which require circuits of super-logarithmic depth?} The question of constructing small-depth circuits is tied not only to the existence of polynomial-size formulas, but also on how easy it is to parallelize the computation of the polynomial. We should note here for comparison that the provenance polynomials for all $\UCQ$s admit $O(\log m)$-depth circuits (and thus polynomial-size formulas) when interpreted over any semiring.

\introparagraph{Our Contributions} We next summarize our contributions in this work.
\begin{itemize}
\item We first show (Section~\ref{sec:bounded}) that logarithmic-depth circuits exist for Datalog programs that are {\em bounded}. These are programs that can reach a fixpoint (over a semiring) in a constant number of iterations, independent of the input size. Thus, bounded programs have polynomial-size formulas. We then extend existing characterizations of boundedness~\cite{Naughton86} from the Boolean semiring to more general semirings. The latter result may be of independent interest.

\item Second, in Section~\ref{sec:rpq} we study Datalog programs with rules that have path bodies: these are called {\em basic chain Datalog programs} and correspond to solving a context-free reachability problem over a labeled graph (reachability via a Context-Free Grammar -- $\CFG$). When the corresponding $\CFG$ is a regular grammar, this class of programs corresponds to Regular Path Queries ($\RPQ$s). For basic chain Datalog, we show a dichotomy: the provenance polynomial has a polynomial-size formula if and only if the corresponding grammar is finite, i.e., accepts finitely many words. For $\RPQ$s, we can show a stonger dichotomy on the depth: circuits can have depth either $\Theta(\log m)$ or  $\Theta(\log m^2)$, with nothing in-between.

\item Third, we show in Section~\ref{sec:upper} a general upper bound on the circuit depth. In particular, we prove that Datalog programs with the {\em polynomial fringe property} admit circuits of depth only $O(\log^2 m)$ for any absorptive semiring. Briefly, the polynomial fringe property says that all proof trees of a fact have polynomially many leaves. All linear Datalog programs have this property, and so all linear programs have small-depth circuits.

\item Finally, in Section~\ref{sec:upper} we show a formula-size dichotomy for monadic connected linear Datalog for semirings that are absorptive and $\otimes$-idempotent. In particular, we show that boundedness is a necessary and sufficient condition for polynomial-size formulas.
\end{itemize}

Before we state and prove our main results, we define preliminary notions in Section~\ref{sec:prelim} and then present some basic results for circuits and formulas in Section~\ref{sec:basic}.
\section{Preliminaries}
\label{sec:prelim}

In this part, we introduce important terminology and background.

\subsection{Datalog}

We follow the notation in~\cite{KhamisNPSW24}. A {\em Datalog} program $\Pi$ consists of a set of rules, {each of the form}:
$$R_0(\bx_0) \obtainedfrom R_1(\bx_1) \wedge \cdots \wedge R_m(\bx_m),$$
where $R_0, \ldots, R_m$ are relation names (not necessarily distinct), and each $\bx_i$ is a tuple of variables and/or constants. The atom $R_0(\bx_0)$ is called the {\em head}, and the conjunction $R_1(\bx_1) \wedge \cdots \wedge R_m(\bx_m)$ is called the {\em body}. A relation, or equivalently predicate, name that occurs in the head of some rule in $\Pi$ is called an {\em intensional database predicate (IDB)}, otherwise it is called an {\em extensional database predicate (EDB)}. {With a slight abuse of notation, we also refer to IDBs and EDBs as relations \em{per se}.} All EDBs form the input database, which is denoted as $I$. For the output we adopt the {\em predicate} I/O convention where a designated IDB predicate, called the {\em target IDB}, is the output~\cite{HillebrandKMV91}. The finite set of all constants occurring in the input $I$ is called the active domain of $I$ and denoted as $\textsf{Dom}(I)$. 
A rule is called an {\em initialization rule} if its body contains no IDBs, and a {\em recursive rule} otherwise.

\introparagraph{Linear and monadic datalog} A rule is called {\em linear} if its body contains at most one IDB and a Datalog program is called {\em linear} if every rule is linear. The {\em arity} of a predicate is the total number of variables and constants it contains. A predicate is called {\em monadic} if it has arity 1, and a Datalog program is called monadic if every IDB is monadic (regardless of the arities of EDBs). 

\introparagraph{Grounding} A {\em grounding} of a Datalog rule is a rule where we instantiate each variable with a constant from the active domain $\textsf{Dom}(I)$. The {\em grounded program} of $\Pi$ consists of all possible groundings for all the rules in $\Pi$.

\begin{example}
Let $E$ be a binary relation that encodes whether two vertices are connected via an edge in a directed graph. Consider the following two Datalog programs:
$$\begin{aligned} & T(x,y) \obtainedfrom E(x,y) \qquad\qquad & U(x) \obtainedfrom A(x)  \\ 
& T(x,y) \obtainedfrom T(x,z) \wedge E(z, y)  \qquad\qquad\qquad & U(x) \obtainedfrom U(y) \wedge E(x,y)
\end{aligned}$$
The first program encodes the transitive closure of the graph ($\TC$) with target the IDB $T$. The second (monadic) program with target the IDB $U$ finds the reachable nodes starting from a set of nodes that satisfy the property $A$.
\end{example}


\introparagraph{Proof Trees} A {\em proof tree} for a fact records how this fact is derived from EDBs and rules in a Datalog program. Formally, we have the following definition.

\begin{definition}[Proof Tree]
Consider a Datalog program $\Pi$, an input database $I$ and a fact $\alpha$ in some IDB of $\Pi$. A proof tree of $\alpha$ w.r.t. $\Pi$ and $I$ is a finite labeled rooted tree $\tau=(V, E, \lambda)$, with $\lambda$ mapping each vertex to a fact, such that:
\begin{enumerate}
\item If $v \in V$ is the root, then $\lambda(v)=\alpha$.
\item If $v \in V$ is a leaf, then $\lambda(v) \in I$.
\item If $v \in V$ is an internal node with children $u_1, \ldots, u_n$, then there is a grounding of a rule of the form $\lambda(v) :- \lambda(u_1), \ldots, \lambda(u_n)$.
\end{enumerate}
\end{definition} 

We further say that a proof tree is {\em tight}~\cite{UllmanG88} if there are no internal nodes $v,v'$ that are in the same leaf-to-root path and $\lambda(v) = \lambda(v')$. Even though the number of a proof trees for a given fact may be infinite, there are only finitely many tight proof trees. 

\begin{example}\label{ex:TC}
    Consider the Datalog program $\TC$ and let $E$ be a binary relation instance shown in Table~\ref{t:E}. Graphically, this corresponds to the directed graph shown in Figure~\ref{fig:E}. A proof tree for the IDB fact $T(s,t)$ is depicted in Figure~\ref{fig:proof_tree}.
    
    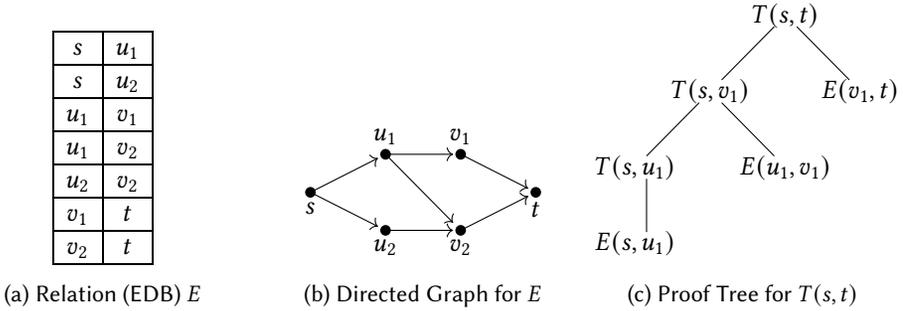
\begin{figure}
    \begin{subfigure}{0.3\linewidth}
    \centering
        \begin{tabular}{|c|c|}
        \hline
           $s$  & $u_1$ \\
        \hline
           $s$  & $u_2$ \\
        \hline
            $u_1$ & $v_1$ \\
        \hline 
            $u_1$ & $v_2$ \\
        \hline
            $u_2$ & $v_2$ \\
        \hline
            $v_1$ & $t$ \\
        \hline
            $v_2$ & $t$ \\
        \hline
        \end{tabular}
        \caption{Relation (EDB) $E$}
        \label{t:E}   
    \end{subfigure}
    \begin{subfigure}{0.3\linewidth}
        \centering
        \begin{tikzpicture}
            \fill (0,0.5) circle[radius=2pt] node[anchor=north]  {$s$};
            \fill (1,1) circle[radius=2pt] node[anchor=south]  {$u_1$};
            \fill (1,0) circle[radius=2pt] node[anchor=north]  {$u_2$};
            \fill (2,1) circle[radius=2pt] node[anchor=south]  {$v_1$};
            \fill (2,0) circle[radius=2pt] node[anchor=north]  {$v_2$};
            \fill (3,0.5) circle[radius=2pt] node[anchor=north] {$t$};
            \draw[->] (0,0.5) -- (0.9,0.95);
            \draw[->] (0,0.5) -- (0.9,0.05);
            \draw[->] (1,1) -- (1.9,1);
            \draw[->] (1,1) -- (1.9,0.1);
            \draw[->] (1,0) -- (1.9,0);
            \draw[->] (2,0) -- (2.9,0.45);
            \draw[->] (2,1) -- (2.9,0.55);
        \end{tikzpicture}
        \caption{Directed Graph for $E$}
        \label{fig:E} 
    \end{subfigure}
    \begin{subfigure}{0.3\linewidth}
        \centering
        \begin{tikzpicture}
            \draw (0,0) node  {$T(s,t)$};
            \draw (-1,-1) node  {$T(s, v_1)$};
            \draw (1,-1) node  {$E(v_1, t)$};
            \draw (-2,-2) node {$T(s,u_1)$};
            \draw (0,-2) node {$E(u_1,v_1)$};
            \draw (-2,-3) node {$E(s,u_1)$};
            \draw (-0.15,-0.15) -- (-0.85,-0.85);
            \draw (0.15,-0.15) -- (0.85,-0.85);
            \draw (-1.15, -1.15) -- (-1.85, -1.85);
            \draw (-0.85, -1.15) -- (-0.15, -1.85);
            \draw (-1.85,-2.15) -- (-1.85,-2.85);
        \end{tikzpicture}
        \caption{Proof Tree for $T(s,t)$}
        \label{fig:proof_tree}
    \end{subfigure}
    \caption{{Example EDB relation and one proof tree for the fact $T(s,t)$ of the Datalog program that computes the transitive closure of a graph ($\TC$). There are two other proof trees for $T(s, t)$.}}
    \end{figure}
\end{example}

\subsection{Semirings}

A (commutative) \emph{semiring} is an algebraic structure $\mathbb{S}:=(\mathbf{D}, \oplus, \otimes, \mathbf{0}, \mathbf{1})$, where $\oplus$ and $\otimes$ are the \emph{addition} and \emph{multiplication} in $\mathbb{S}$ such that: (1) $(\mathbf{D}, \oplus, \mathbf{0})$ and $(\mathbf{D}, \otimes, \mathbf{1})$ are commutative monoids, (2) $\otimes$ is distributive over $\oplus$, (3) $\mathbf{0}$ is an annihilator of $\otimes$ in $\mathbf{D}$.  Notable examples are the Boolean semiring $\bB:=$ ( $\{\textsc{False},\textsc{True}\}, \vee, \wedge$, \textsc{False}, \textsc{True} $)$, the Tropical semiring $\mathbb{T}:=(\mathbb{N}\cup \{+\infty\}, \min ,+,+\infty, 0)$, the $\text{Tropical}^-$ semiring $\mathbb{T}^-:=(\mathbb{Z}\cup \{+\infty\}, \min ,+,+\infty, 0)$, and the Counting semiring $\mathbb{C}:=(\mathbb{N},+, \cdot, 0,1)$.\footnote{$\mathbb{N}$ stands for the natural numbers, and $\mathbb{Z}$ for integers.}

We say that $\mathbb{S}$ is {\em idempotent} if $x \oplus x = x$ for every $x \in \mathbf{D}$ and that $\mathbb{S}$ is {\em absorptive} if $\mathbf{1} \oplus x = \mathbf{1}$ for every $x \in \mathbf{D}$. An absorptive semiring is idempotent, since $x \oplus x = x\otimes (\mathbf{1}+\mathbf{1}) = x\otimes \mathbf{1} = x$ for any $x \in \mathbf{D}$, not necessarily the other way around. For example, $\bT$ is absorptive while $\bT^-$ is idempotent but not absorptive.
A semiring $\bS$ is {\em positive} if the mapping $h: \bS \rightarrow \bB$ defined by $h(x):= \textsc{False}$ if $x = \mathbf{0}$ and $h(x):= \textsc{True}$ if $x \neq \mathbf{0}$ is a semiring homomorphism\footnote{That is, the mapping respects the $\otimes$ and $\oplus$ operations in the semirings.}~\cite{Green11}. All semirings discussed in this paper will be positive semirings. 

Finally, we say that a semiring $\bS$ is {\em naturally ordered} if the binary  relation $x \leq_\bS y$ defined as $\exists z: x \oplus z =y$ is a partial order. The semirings $\bB, \mathbb{T}, \mathbb{C}$ are all naturally ordered.


\subsection{Datalog Semantics over Semirings}

{We provide a short introduction to Datalog semantics over semirings; interested readers should consult~\cite{KhamisNPSW24} and references therein for a more formal definition. For a fixed program and an input database,} let $N$ denote the total number of IDB facts and let $\bS$ be a naturally ordered semiring.  We now define the immediate consequence operator (ICO) as a map $f: \bS^N \rightarrow \bS^N$ that maps each IDB fact $\alpha$ to a new value by $\oplus$-summing over all grounded rules with $\alpha$ as the head, and for each such rule taking the $\otimes$-product of its facts. If the semiring $\bS$ is naturally ordered, then the least fixpoint is the least fixed point of $f$ under the same partial order extended to $\bS^n$ componentwise. To compute this, we can apply {\em naive evaluation}: we start with $x \gets \mathbf{0}$ for every IDB fact, and repeatedly apply $f$  until a fixpoint is reached. 

This process may not converge in general. However, it always converges in a finite number of steps if the semiring is $p$-stable~\cite{KhamisNPSW24}. A semiring is {\em $p$-stable} for some integer $p \geq 0$ if for any element $u$ in the semiring, $\mathbf{1} \oplus u \oplus u^2 \oplus \dots \oplus u^p = \mathbf{1} \oplus u \oplus u^2 \oplus \dots \oplus u^{p+1}$. Note that an absorptive semiring is exactly a 0-stable semiring.

\subsection{Provenance Polynomials}\label{sec:provenance_polynomial}

In this paper, we follow the convention in provenance polynomials and consider each EDB fact to be {\em tagged by a variable of an semiring element}~\cite{GreenKT07}. That is, we associate each EDB fact $\alpha$ with a variable $x_\alpha$ representing the value of that fact over a semiring $\bS$. Given a proof tree $\tau$ with leaves $\mathcal{L}(\tau)$, the {\em monomial of a proof tree} is the product of all variables of the leaves in that proof tree. The {\em provenance polynomial of a fact} is  the sum of  monomials of all possible proof trees for the fact:
$$ \bigoplus_{\tau:\text{ proof tree}} \bigotimes_{t \in \mathcal{L}(\tau)} x_{\lambda(t)}$$
There could be infinitely many proof trees for a fact. However, if we want to interpret the polynomial only over absorptive semirings, it is equivalent to consider only the (finitely many) tight proof trees. This gives us the following finite polynomial:
$$ p^I_\Pi(\alpha) := \bigoplus_{\tau:\text{ tight proof tree}} \bigotimes_{t \in \mathcal{L}(\tau)} x_{\lambda(t)}$$

\begin{proposition} \label{prop:tight:tree}
    Consider any Datalog program $\Pi$, input database $I$ and an IDB fact $\alpha$. Let $\bS$ be any absorptive semiring. Then, $p^I_\Pi(\alpha)$ is equivalent to $ \bigoplus_{\tau:\text{ proof tree}} \bigotimes_{t \in \mathcal{L}(\tau)} x_{\lambda(t)}$ over $\bS$.
\end{proposition}
\begin{proof}
To show this, we will prove that every monomial that corresponds to a non-tight proof tree gets absorbed by the monomial of a tight proof tree. Indeed, observe that any internal node of a proof tree corresponds to an IDB fact. If a proof tree is not tight, then there exist two internal nodes $n_1, n_2$ along a leaf-to-root path sharing the same IDB fact. Without loss of generality, assume that $n_2$ is in the subtree of $n_1$. We can then replace the subtree of $n_1$ by the subtree of $n_1$. Clearly, the resulting tree is a proof tree of $\alpha$ and, more importantly, its leaves is a strict sub-multi-set of the original proof tree's. We repeat this process until we end up with a tight proof tree for $\alpha$. {The constructed monomial of a tight proof tree absorbs the monomial of the given non-tight proof tree due to the algebraic equation $m \oplus m \otimes n  = m \otimes (\mathbf{1}\oplus n) = m \otimes \mathbf{1} = m$ where the second equality applies the absorption rule.}
\end{proof}

As an example, the provenance polynomial of a fact in $\TC$ over the Tropical semiring $\bT$ computes the minimum of the weight of all paths for a particular pair of vertices, where the weight of a path is the sum of all its edge weights. The provenance polynomial for $T(s,t)$ in Example~\ref{ex:TC} is 
$$
p^{E}_{\TC}(T(s,t)) = (x_{s,u_1} \otimes x_{u_1, v_1} \otimes x_{v_1,t}) \oplus (x_{s,u_1} \otimes x_{u_1, v_2} \otimes x_{v_2,t}) \oplus (x_{s, u_2} \otimes x_{u_2, v_2} \otimes x_{v_2, t}).
$$
\subsection{Circuits and Formulas over Semirings for Datalog}

A \emph{circuit} $F$ over a semiring $\mathbb{S}$ is a Directed Acyclic Graph (DAG) with input nodes (with fan-in 0) variables representing EDB facts in a set $S_x$ and the constants $\mathbf{0}, \mathbf{1}$. Every other node is labelled by $\oplus$ or $\otimes$ and has fan-in 2; these nodes are called \emph{$\oplus$-gates} and \emph{$\otimes$-gates}, respectively. An input gate of $F$ is any gate with fan-in 0 and an output gate of $F$ is any gate with fan-out 0. The {\em size} of the circuit $F$, denoted as $|F|$, is the number of gates in $F$. The {\em depth} of $F$ is the length of the longest path from an input to an output node. 

A {\em formula} over a semiring $\mathbb{S}$ is a circuit over $\mathbb{S}$ such that the outdegree of every gate is exactly one (with the exception of the output gate, which has degree 0). In other words, the difference between a formula and a circuit is that any gate output cannot be reused across multiple gates. The depth of a formula is defined analogously. 

A {\em polynomial over a semiring} $\bS$ is a finite expression involving $\oplus$ and $\otimes$ of variables representing elements in $\bS$. Its canonical (or DNF) form is rearranging the expression to be a sum $\oplus$ of product $\otimes$ of variables. Two polynomials $p_1$ and $p_2$ are said to be {\em equivalent over a semiring} $\bS$, denoted as $p_1 \equiv_{\bS} p_2$, if $p_1$ and $p_2$ define the same function over $\bS$. {By Proposition~\ref{prop:tight:tree}, a provenance polynomial over any absorptive semiring $\bS$ has an equivalent polynomial over $\bS$, and thus the corresponding function over $\bS$ is well-defined.}

A circuit (or a formula) over $\bS$ is said to {\em produce} a polynomial over $\bS$ if its output gate coincides with the polynomial as its canonical form (up to rearranging the ordering of monomials and variables therein) when evaluating the circuit in the bottom-up fashion. A circuit (or a formula) is said to {\em compute} a polynomial $p$ over $\bS$ if the polynomial it produces is equivalent to $p$ over $\bS$.

We take the point of view that a polynomial over a semiring is a formal expression and its $\oplus$ and $\otimes$ operations are subject to different interpretations over different semirings. 

\section{Known Bounds on Circuits/Formulas for Datalog over Semirings}
\label{sec:basic}

In this section, we present some basic lower and upper bounds that are known or can be easily derived from the state-of-the-art.

For a general Datalog program $\Pi$ over an absorptive semiring $\bS$ and an IDB fact $\alpha$, it could be that the number of proof trees is exponential in the input size. Hence, the provenance polynomial $p^I_\Pi(\alpha)$, written in DNF, could be exponentially large. But, it was shown~\cite{DeutchMRT14} that a circuit representation of $p^I_\Pi(\alpha)$ is always of polynomial size. 

\begin{theorem}[Deutch et. al~\cite{DeutchMRT14}]
Let $\Pi$ be a Datalog program and $\bS$ be any absorptive semiring. Then, for any input $I$ and any IDB fact $\alpha$, we can construct a circuit for the provenance polynomial $p^I_\Pi(\alpha)$ over $\bS$ with polynomial size in $I$.
\end{theorem}

In particular, the size of the above constructed circuit is the size of the grounded program, and the depth is the number of IDB facts in the grounded program times a logarithmic factor.
The question on whether we can construct a polynomial-size formula for the provenance polynomial is more complicated. Interestingly, the question about the minimum size of the formula is tied to finding the minimum depth of a circuit, as the following result shows.

\begin{theorem}[Wegener~\cite{Wegener83}]\label{thm:size_depth_tradeoff}
Let $F$ be a formula over the Boolean semiring of size $|F|$. Then, there exists an equivalent formula (circuit) of depth $O(\log |F|)$.
\end{theorem}

In other words, to show super-polynomial lower bounds for formulas over the Boolean semiring, it suffices to show that circuits have super-logarithmic depth. For the other direction, we have the following folklore proposition:

\begin{proposition}
Let $F$ be a circuit over any semiring of depth $d$. Then, there exists an equivalent formula of size $2^d$ and depth $d$.
\end{proposition}

This means that circuits of logarithmic depth imply polynomial-size formulas. Combined, we obtain that a Boolean function has a formula of size polynomial in the input size $n$ if and only if it has circuits of depth $O(\log n)$ over $\mathbb{B}$. 

Hence, we can apply existing depth lower bounds for our purposes. A celebrated result by Karchmer and Widgerson~\cite{KarchmerW90} asserts that any monotone Boolean circuit that decides whether $t$ is reachable from $s$ in a directed graph with $n$ vertices must have depth $\Omega(\log^2 n)$. We formally state this result in the language of provenance polynomials. Recall that $\TC$ is the Datalog program for transitive closure. Moreover, define a {\em $(\ell, n)$-layered directed graph} to be a directed graph with $n$-layers, each layer having $\ell$ vertices, such that there only exist directed edges between two vertices that are in different but consecutive layers.

\begin{theorem}[Karchmer \& Wigderson~\cite{KarchmerW90}]\label{thm:Karchmer_Widgerson}
For any $n>0$, there exists an input $I$ that is a $(n^{0.1}, n)$-layered directed graph and a fact $\alpha = T(s,t)$ such that any circuit that computes the provenance polynomial $p^I_{\TC}(\alpha)$ over $\bB$ requires $\Omega(\log^2 n)$ depth.
\end{theorem}

We now observe that $st$-connectivity over a layered directed graph actually admits a linear size circuit. In other words, Theorem~\ref{thm:Karchmer_Widgerson} demonstrates the trade-off between linear-size circuit and superpolynomial-size formula.

\begin{theorem}
    Over any semiring $\bS$, the provenance polynomial for $st$-connectivity over a layered directed graph admits a linear size circuit with linear depth.
\end{theorem}
\begin{proof}
    Essentially, we prove that the graph itself is a circuit for $st$-connectivity. We construct the circuit inductively in a bottom-up fashion. For the base case, any edge is computed by the input edge variable that represents it. For any internal vertex, if there are more than 1 incoming edge, we first construct a $\oplus$-gate that sums over the polynomials for those edges. We say this $\oplus$-gate computes the polynomial for this vertex. Then, for each outgoing edge of that vertex, construct a $\otimes$-gate between the polynomial for this vertex and the edge variable of the outgoing edge. Clearly, this circuit computes $st$-connectivity correctly and is of linear size and linear depth.
\end{proof}

In the above theorem the two distinguished vertices $s$ and $t$ are located at the bottom layer and the top layer respectively (c.f. Theorem 4.7 in~\cite{BarakKRVV23}). The particular structure of the layered graphs will be important for our lower bound results. A circuit with depth matching this lower bound can be constructed using the ``path-doubling'' technique~\cite{KarchmerW90,UllmanG88}.

An immediate consequence of~\autoref{thm:Karchmer_Widgerson} and Theorem~\ref{thm:size_depth_tradeoff} is that the Boolean provenance polynomials for the $\TC$ program need in general super-polynomial size formulas. Observe that the above lower bound applies only to circuits over the Boolean semiring, so one might ask what happens for other semirings. Fortunately, lower bounds can "transfer up" to any positive semiring.

\begin{proposition}\label{prop:positive_semiring}
    Any size or depth upper bound for circuits or formulas over any positive semiring $\bS$ for a polynomial $p$ over $\bS$ implies the same upper bound for $p$ over $\bB$. Any size or depth lower bound for circuits or formulas for $p$ over $\bB$ implies the same lower bound for $p$ over any positive semiring $\bS$.
\end{proposition}
\begin{proof}
Given a circuit (or formula) over $\bS$ for $p$, we simply interpret each $\oplus$-gate as $\wedge$ and each $\otimes$-gate as $\vee$, and consider each variable to take values from $\{\textsc{False}, \textsc{True}\}$. The resulting circuit (or formula) computes $p$ over $\bB$ since $\bS$ is positive, and thus there is a semiring homomorphism from $\bS$ to $\bB$.
\end{proof}

Finally, one might wonder for which Datalog programs it would be possible to construct circuit of logarithmic depth -- and thus formulas of polynomial size. One first observation is that Datalog programs with only initialization rules correspond to $\UCQ$s; but in this case, the possible proof trees are only polynomially many.

\begin{proposition}\label{thm:UCQ_AC0}
There exists a circuit of size $\operatorname{poly}(|I|)$ and depth $O(\log |I|)$ for any $\UCQ$ with input $I$ over any semiring. Moreover, there exists a formula of size $\operatorname{poly}(|I|)$ for any $\UCQ$ with input $I$ over any semiring.
\end{proposition}

\begin{proof}
The provenance polynomial of any $\UCQ$ over any semiring is simply the sum over all valid valuations of the product over tuples, for every rule with the same IDB head. The circuit simply takes the sum of product form, where we do the $\oplus$-gates and $\otimes$-gates in a binary tree fashion. Since there are at most $\operatorname{poly}(|I|)$ many summands, the depth of the circuit is at most $O(\log |I|)$.
\end{proof}

As we will see in the next section, the above proposition can be generalized to other Datalog programs that we call {\em bounded}.
\section{Boundedness of Datalog over Semirings}
\label{sec:bounded}

In this section, we define the notion of boundedness for a Datalog program over a semiring, and then explore different characterizations of boundedness. Boundedness is a fundamental property for our purposes because it implies small-depth circuits.

\begin{definition}[Boundedness]
A Datalog program $\Pi$ with target a predicate $T$ over a naturally ordered semiring $\bS$ is called \emph{bounded} if there exists a constant $k$ (independent of the input size) such that the naive evaluation of $\Pi$ reaches the fixpoint for $T$ in at most $k$ iterations for any input.   
\end{definition}

It could be the case that some IDB predicates in the program are bounded, but the program itself is not. In general, the problem of whether a given Datalog program $\Pi$ is bounded for the Boolean semiring is undecidable~\cite{HillebrandKMV91, GaifmanMSV93, Vardi88, Naughton89, NaughtonS87}.  For example, it is shown that program boundedness is undecidable for two linear rules and one initialization rule~\cite{HillebrandKMV91}. 


\begin{example}
Consider the following Datalog program:
$$\begin{aligned} 
& T(x,y) \obtainedfrom E(x,y) \\ 
& T(x,y) \obtainedfrom A(x) \wedge T(z,y)
\end{aligned}$$
This program is bounded for any absorptive semiring. Indeed, in this case the program is equivalent to one where we replace the second rule with $T(x,y) \obtainedfrom A(x) \wedge E(z,y) $, which is a $\UCQ$.  
\end{example}

\introparagraph{Boundedness and Circuits}
Boundedness is a crucial property that allows us to construct semiring circuits of optimal depth, i.e., logarithmic to the input size, as the next theorem shows.

\begin{theorem} \label{thm:bounded:depth}
Let $\Pi$ be a bounded Datalog program with target $T$ over a naturally ordered semiring $\bS$. Then, for any input $I$ of size $n$ and any fact $\alpha$ of $T$, we can construct a circuit for the provenance polynomial $p^I_\Pi(\alpha)$ over $\bS$ with polynomial size and depth $O(\log |I|)$.
\end{theorem} 

\begin{proof}
Since the predicate $T$ is bounded, there is a constant $k$ independent of the input size such that after $k$ fixpoint iterations every fact of $T$ has reached a fixpoint. This means that we can construct a circuit with $k$ layers, where each layer encodes the evaluation of one iteration.

To construct each layer, we first construct the grounded program w.r.t. $\Pi$ and the input $I$. Then, we compute every grounded rule -- this has constant depth, since every grounded rule can be computed via a number of $\otimes$-operations linear to the size of the body of the rule. Finally, we need to $\oplus$-sum all grounded rules with the same IDB fact as head. 

We now argue about the size and depth of each layer. First, note that the number of gates in each layer is $O(M)$, where $M$ is the size of the grounded program. $M$ is always polynomial to the input size. Second, for every IDB fact $\alpha$, there can be at most a polynomial number of grounded rules with head $\alpha$. Since we can always implement a commutative and associative summation with a circuit of logarithmic depth, we have depth only $O(\log |I|)$. 
\end{proof}


\introparagraph{Characterizing Boundedness}
When Datalog is evaluated over the Boolean semiring, boundedness is equivalent to saying that the target predicate is equivalent to a $\UCQ$~\cite{Naughton86}. We will show that we can extend this equivalence to a subclass of absorptive semirings.

To prove this, we will use the fact that for a predicate $T$ of $\Pi$ in a $p$-stable semiring $\bS$, there exists an infinite sequence of $\CQ$s $C_0, C_1, C_2, \dots$ such that for every input database $I$, we have $T(I) =_{\bS} \bigcup_{i=0}^\infty C_i(I)$. The $\CQ$s $C_0, C_1, \dots$ are called the {\em expansions} of $\Pi$ for $T$. The $p$-stability is necessary such that $T(I)$ is a well-defined quantity.

\begin{example}
Consider the transitive closure example. The expansions for $T$ can be defined as:
\begin{align*}
    C_0(x,y) & := E(x,y) \\
    C_1(x,y) & := E(x,z) \wedge E(z,y) \\
    C_2(x,y) & := E(x,z) \wedge E(z,w) \wedge E(w,y) \\
    \dots
\end{align*}
\end{example}

We will also need to define containment of $\UCQ$s w.r.t. a semiring $\bS$, denoted as $\subseteq_\bS$. To do this, we will use the natural order of the semiring $\bS$. To guarantee that this exists, we will consider $\oplus$-idempotent semirings, where natural order is defined as $a \leq_\bS b \Leftrightarrow a \oplus b = b$. (see~\cite{KRS14} for a formal treatment of $\UCQ$s containment w.r.t. any semiring.)

\begin{theorem} \label{thm:bound:idem}
 Let $T$ be the target predicate for a Datalog program $\Pi$ over an idempotent $p$-stable semiring $\bS$. Let $C_0, C_1, \dots$ be its expansions. Then, $\Pi$ is bounded if and only if there exists an integer $N \geq 1$ such that for all $n > N$, $C_n \subseteq_{\bS} \bigcup_{i=0}^{{N}} C_i$.    
\end{theorem}

\begin{proof}
First, assume that the condition holds.
Consider some fact $\alpha$ in $T$, and consider its provenance polynomial for some instance $I$. We note that we can write its polynomial as $\sum_{i=0}^\infty p^I_{C_i}(\alpha)$, where $p^I_{C_i}(\alpha)$ is the provenance polynomial of $\alpha$ for the expansion CQ $C_i$. Since the semiring is $p$-stable, this polynomial is well-defined (i.e., has a finite value), and moreover there is some $K$ (possibly depending in the instance $I$) such that $\sum_{i=0}^\infty p^I_{C_i}(\alpha) = \sum_{i=0}^K p^I_{C_i}(\alpha)$.

Then, we have that for any assignment $\phi$ of the polynomial to the semiring domain:
\begin{align*}
\phi(p^I_{\Pi}(\alpha)) = \bigoplus_{i=0}^\infty \phi(p^I_{C_i}(\alpha)) 
& = \left( \bigoplus_{i=0}^n \phi(p^I_{C_i}(\alpha)) \oplus \phi(p^I_{C_{n+1}}(\alpha)) \right) \oplus \bigoplus_{i=n+2}^K \phi(p^I_{C_i}(\alpha)) \\
& = \left( \bigoplus_{i=0}^n \phi(p^I_{C_i}(\alpha) ) \right) \oplus \bigoplus_{i=n+2}^K \phi(p^I_{C_i}(\alpha)) \\
& = \left( \bigoplus_{i=0}^n \phi(p^I_{C_i}(\alpha) ) \right) \oplus \bigoplus_{i=n+3}^K \phi(p^I_{C_i}(\alpha)) \\
& = \dots \\
& = \bigoplus_{i=0}^n \phi(p^I_{C_i}(\alpha) )
\end{align*}
where the second equality (and all subsequent ones) holds because $C_{n+1} \subseteq_{\bS} \bigcup_{i=0}^n C_i$ and thus $\phi(p^I_{C_{n+1}}(\alpha)) \leq_{\bS} \bigoplus_{i=0}^n \phi(p^I_{C_i}(\alpha))   $.

For the second direction, assume that $T$ is bounded and reaches the fixpoint in at most $k$ iterations. Then, by taking the CQ expansions corresponding to the first $k$ iterations, there exists $n>0$ such that for any input $I$, $T(I) =_{\bS} \bigcup_{i=1}^n C_i(I)$. This means that for any input $I$ and fact $\alpha$, we have that
$$ \bigoplus_{i=0}^n \phi(p^I_{C_i}(\alpha)) = \bigoplus_{i=0}^\infty \phi(p^I_{C_i}(\alpha)) $$
Now, note that for any $j > n$:
$$\phi(p^I_{C_j}(\alpha)) \leq_{\bS} \bigoplus_{i=0}^\infty \phi(p^I_{C_i}(\alpha)) )$$
and thus we have the desired containment property.
\end{proof}

When the semiring is absorptive and $\otimes$-idempotent, we can obtain an even stronger characterization that is equivalent to that of the Boolean semiring. We note that absorptive and $\otimes$-idempotent semirings correspond to the class $\mathbf{C}_{\textsf{hom}}$ of semirings, defined in~\cite{KRS14}. 
{Naaf showed that $\mathbf{C}_{\textsf{hom}}$ is exactly the class of bounded distributive lattices (Proposition 3.1.8 in~\cite{Naaf24}).} For this class, $\CQ$ containment is exactly characterized by the existence of a homomorphism.

\begin{theorem}
Let $T$ be the target predicate for a Datalog program $\Pi$  over an $\otimes$-idempotent and absorptive semiring $\bS$. Let $C_0, C_1, \dots$ be its expansions. Then, $\Pi$ is bounded if and only if there exists $N \geq 1$ such that for all $n > N$, there exists $m \leq N$ such that there is a homomorphism from $C_m$ to $C_n$.   
\end{theorem}

\begin{proof}
Indeed, we know from~\cite{KRS14} that for any such semiring, for $\UCQ$s $Q_1, Q_2$, we have that $Q_1 \subseteq_\bS Q_2$ if and only if for every $\CQ$ $q$ in $Q_1$ there exists a $\CQ$ $q'$ in $Q_2$ such that there is a homomorphism from $q'$ to $q$. Combined with Theorem~\ref{thm:bound:idem}, we obtain the desired result.
\end{proof}

The above condition for boundedness coincides precisely with the condition of boundedness for the Boolean semiring. Hence, the property of boundedness with absorptive $\otimes$-idempotent semirings is completely characterized by the behavior of boundedness in the Boolean semiring.

\begin{corollary} \label{cor:bounded:equiv}
 Let $\bS$ be an absorptive $\otimes$-idempotent semiring, and $\Pi$ be a Datalog program. Then, $\Pi$ is bounded over $\bS$ if and only if it is bounded over the Boolean semiring.  
\end{corollary}

Finally, we obtain the following equivalent characterization of boundedness, which was only known for the Boolean semiring.

\begin{proposition} \label{cor:bounded:ucq}
 Let $\bS$ be an absorptive $\otimes$-idempotent semiring, and $\Pi$ be a Datalog program. Then, $\Pi$ is bounded over $\bS$ if and only if it the target predicate $T$ is equivalent to a $\UCQ$.
\end{proposition}

\begin{proof}
For the first direction, suppose that $T$ is bounded. Then, the construction in Theorem~\ref{thm:bound:idem} tells us that $T$ is equivalent to a finite union of CQ expansions, which is a UCQ.

For the other direction, suppose that the IDB predicate $T$ is equivalent to the UCQ $\cup_{j=1}^\ell D_j$. Consider the expansions of $T, \cup_{i=0}^\infty C_i$. Using again the UCQ containment characterization in~\cite{KRS14}, since $\cup_{j=1}^\ell D_j \subseteq_\bS \cup_{i=0}^\infty C_i$, for every $D_j$, there exists a homomorphism from some $C_{i_j}$ to $D_j$. Consider now the UCQ $\cup_{j=1}^\ell C_{i_j}$. Since $ \cup_{i=0}^\infty C_i \subseteq_\bS \cup_{j=1}^\ell D_j$, for any $C_i$, there exists a homomorphism $D_{j} \rightarrow C_i$ for some $D_j$, and hence a homomorphism $C_{i_j} \rightarrow C_i$. Hence, 
$\cup_{j=1}^\ell C_{i_j} =_{\bS} \cup_{i=0}^\infty C_i =_{\bS} T$. But this means that if we take the largest index $k = \max_{j=1,\ell} i_j$, the fixpoint will be reached after at most that many iterations.
\end{proof}

We do not know whether the characterization of boundedness via equivalence to a $\UCQ$ holds for more general semirings than the ones in the class $\mathbf{C}_{\textsf{hom}}$. It is also an open question whether there is a program $\Pi$ that is bounded over the Boolean semiring, but is unbounded over an absorptive one (even, say, the tropical semiring).

\section{Circuit Bounds for Basic Chain Datalog}
\label{sec:rpq}

A {\em chain rule} is a rule of the following form:
$$
P(x, y) \obtainedfrom Q_0(x, z_1) \wedge Q_1(z_1, z_2) \wedge \ldots \wedge Q_k(z_k, y),
$$
where the $Q_i$'s are binary predicates and $x, y, z_i$'s are distinct variables. A {\em basic chain Datalog} program is a program $\Pi$ whose recursive rules are chain rules~\cite{Yannakakis90, UllmanG88}. A basic chain Datalog program $\Pi$ corresponds to a {\em context-free grammar} ($\CFG$) in the following way: IDB and EDB predicates correspond to non-terminal symbols and terminal symbols in the $\CFG$ respectively, with the designated IDB predicate as the starting non-terminal symbol, and rules in $\Pi$ corresponds to productions in the $\CFG$ by ignoring the variables. For example, the Datalog program $\TC$ corresponds to the $\CFG$ $T \gets T E \mid E$. 

\begin{definition}[Context-Free Reachability] 
Given a directed edge-labeled graph $G$ with labels over a finite alphabet $\Sigma$ and a $\CFG$ $L$, the context-free reachability problem returns all pairs of vertices $(s,t)$ such that there exists a path whose concatenation of its edge labels lies in $L$.
\end{definition}

A regular language is equivalent to a left-linear (or right-linear) $\CFG$, and thus a basic chain Datalog program where all rules are left-linear (or all are right-linear) corresponds to a regular language~\cite{Yannakakis90}. The class of basic chain Datalog programs that correspond to regular languages captures exactly the class of Regular Path Queries ($\RPQ$s)~\cite{AV97}. An $\RPQ$ is a context-free reachability problem where the language $L$ is a regular one. For example, the $\TC$ program corresponds to the $\RPQ$ with alphabet $\{E\}$ and regular language $E^*$.
The next proposition states the connection between $\CFG$s, regular grammars and Datalog programs in a formal manner.  


\begin{proposition}\label{prop:datalog_from_RPQ}
Context-free reachability is equivalent to a basic chain Datalog program such that all rules are of the form $T(x,y) \obtainedfrom U(x,z), V(z,y)$ or $T(x,y) \obtainedfrom U(x,z)$. Additionally, if the language is regular, all recursive rules are left-linear.
\end{proposition}



In this section, we prove both upper bound and lower bound results on the size and depth of circuits and formulas for basic chain Datalog programs. This is not only an interesting result on its own, but also demonstrates the main ideas which will be applied to more complex Datalog programs later on. The results of this section are summarized in Table~\ref{tab:regular}. In particular, we can show the following two dichotomies, that depend on whether the corresponding grammar accepts finitely many words.

\begin{theorem}
Let $\Pi$ be a basic chain Datalog program with target $T$ that corresponds to a regular language $L$ over an absorptive semiring $\bS$. For input size $m$:
\begin{itemize}
\item if $L$ is finite, $T$ has $\Theta(\log m)$-depth circuits and polynomial-size formulas;
\item  if $L$ is infinite, $T$ has $\Theta(\log^2 m)$-depth circuits and super-polynomial-size formulas.
\end{itemize}
\end{theorem}

\begin{theorem}
Let $\Pi$ be a basic chain Datalog program that corresponds to a $\CFG$ $L$ over an absorptive semiring $\bS$. Then, $T$ admits polynomial-size formulas if and only if $L$ is finite.  
\end{theorem}

The key idea is that boundedness for basic chain Datalog is exactly characterized by the finiteness of the grammar $L$ for any absorptive semiring. On the other hand, for infinite grammars, we can reduce from the transitive closure, hence getting a $\Omega(\log^2 m)$ depth lower bound.

\renewcommand{\arraystretch}{1.25}
\begin{table}[t]
    \centering
    \begin{tabular}{c|c|c|c|c}
    \textbf{$\CFG$ Language} & \multicolumn{2}{|c|}{\textbf{Circuit Size}} & \multicolumn{2}{|c}{\textbf{Circuit Depth}} \\
    \hline
    & Upper Bound & Lower Bound & Upper Bound & Lower Bound  \\ \hline
    finite  & $O(m)$ & $\Omega(m)$ & $O(\log n)$ & $\Omega(\log n)$   \\
    \hline
    \multirow{2}{3em}{infinite  regular} & $O(m n)$ & \multirow{2}{2em}{$\Omega(m)$} & $O(n \log n)$ & \multirow{2}{4em}{$\Omega(\log^2 n)$} \\
    & $O(n^3 \log n)$ & & $O(\log^2 n)$ &  \\
    \hline 
    infinite & $O(n^5)$ & $\Omega(m)$ & $O(n^2 \log n)$ & $\Omega(\log^2 n)$ \\
    \end{tabular}
    \caption{Depth and Size for (absorptive) semiring circuits for basic chain Datalog. Here, $m$ is the input size, and $n$ is the size of the active domain. }
    \label{tab:regular}
\end{table}

\subsection{Boundedness for Basic Chain Datalog}

As a first step, we will show that boundedness is characterized by the finiteness of the $\CFG$. Note that this is not implied by our results in the previous section, as this result holds for any absorptive semiring -- for example, it holds for the tropical semiring as well.

\begin{proposition}\label{prop:RPQ_bounded_iff_Datalog_bounded}
Let $\Pi$ be a basic chain Datalog program with target IDB $T$ over an absorptive semiring $\bS$. Then, $T$ is bounded if and only if the corresponding $\CFG$ $L$ of $\Pi$ is finite.
\end{proposition}

\begin{proof}
For the one direction, suppose that $L$ is a finite $\CFG$. This implies that $L$ only accepts words with length up to some constant $k>0$. Note that at iteration $i>0$ the naive evaluation of $\Pi$ has discovered all paths between two nodes $s,t$ of length at least $i$ that satisfy $L$. This means that the evaluation will reach the fixpoint after at most $k$ iterations, since no path with length $>k$ can be recognized by $L$. Note that this direction holds for any semiring. 

For the other direction, suppose that $L$ is infinite. Then, for any $k>0$, there exists a word $w_k$ with length more than $k$ that is accepted by $L$. Consider now the inputs $\{I_k
\}_{k>0}$ that consist of a path of $|w_k|$ edges with labels that form the word $w_k$.  Note that for any basic chain Datalog program $\Pi$ that computes $L$, at iteration $i>0$ naive evaluation has discovered the paths of length at most $c^i$ for some constant $c$. Suppose we run $\Pi$ on $I_k$ over the Boolean semiring; in this case, naive evaluation would need at least $\log_c k$ rounds to reach a fixpoint, because it needs to find the path of length $k$. For any positive semiring $\bS$ the same lower bound holds because of the semiring homomorphism to the Boolean semiring. 
\end{proof}

In addition, since finiteness of a $\CFG$ can decided in polynomial time, we can decide whether a basic chain Datalog program is bounded or not efficiently.

\subsection{Upper Bounds for Regular Path Queries}

We start with an upper bound result for $\TC$ over an arbitrary absorptive semiring.

\begin{theorem}\label{thm:ssst_reachability}
The provenance polynomial of any fact of the $\TC$ program over any absorptive semiring $\bS$ can be computed by a circuit of size $O(mn)$ and depth $O(n \log n)$,  where $m$ is the input size and $n$ is the size of the active domain.
\end{theorem}


\begin{proof}
The main observation is that the Bellman-Ford algorithm is monotone, computes single-source single-target reachability over any absorptive semiring, and adapts to sparse graphs~\cite{Bellman58, Ford56, Jukna15}. Assume without loss of generality the source is vertex 1 and the target is vertex $n$. Let $x_{i,j}$ be the variable that corresponds to the edge $E(i,j)$. We will compute the polynomials $f_j^k$, where $1 \leq j \leq n$ and $1\leq k\leq n-1$, for the $k$-th layer and the $j$-th vertex by the recursion 
$$f_j^k := f_j^{k-1} \oplus  \bigoplus_{i \in N_j} (f_{i}^{k-1} \otimes x_{i,j}),$$
where $N_j$ is the in-neighborhood of vertex $j$.
 The initialization consists of $f_j^1 := x_{1,j}$ (if the edge does not exist, we set it to $\mathbf{0}$) and the output gate is simply $f_n^{n-1}$. Clearly each $\otimes$-gate has fan-in 2. We can implement the summation $\bigoplus_{i \neq j}$ with $(|N_j|-1)$
 $\oplus$-gates of fan-in 2 organized in a binary tree of depth $O(\log n)$.

The resulting circuit has $n-1$ layers, each of depth $O(\log n)$ and thus the circuit depth is $O(n \log n)$. Additionally, each layer has $O(n+\sum_j |N_j|) = O(n+m)$ gates. Assuming $m \geq n$ (otherwise we can simply ignore isolated vertices), the circuit has size $O(mn)$. 
 
 We now argue this circuit correctly computes $\TC$ over an absorptive semiring. Indeed, the output gate computes the polynomial which is the sum over all {\em walks}\footnote{A walk can go over the same vertex multiple times.} of length no greater than $n$ from $1$ to $j$ of the product over all edges within that walk. Formally, it produces the polynomial 
    $$
    \bigoplus_{\substack{\text{walk $w$ from $1$ to $n$} \\ w = i_1, i_2, \dots, i_k  \\ i_1 = 1, i_k = n, k \leq n}}\bigotimes_{0 \leq j \leq k-1} x_{i_{j}, i_{j+1}}
    $$
 In particular, the monomial of every path from $1$ to $n$ is produced, i.e., it appears in the polynomial exactly as it is. Since $\bS$ is absorptive, every monomial that corresponds to a walk that is not a path will be absorbed by a monomial that corresponds to a path as a subset of the walk. Hence, the above polynomial is equivalent to:
    $$
    \bigoplus_{\substack{\text{path $p$ from $1$ to $n$} \\ p = i_1, i_2, \dots, i_k \\ i_1 = 1, i_k = n, k \leq n}}\bigotimes_{0 \leq j \leq k-1} x_{i_{j}, i_{j+1}}
    $$
    due to the absorption. This finishes the proof.
\end{proof}

When the input has size $m = \Theta(n^2)$, i.e., the input graph is dense, the circuit constructed in Theorem~\ref{thm:ssst_reachability} has size $O(n^3)$ and depth $O(n\log n)$. We show next that one can construct a circuit of size $O(n^3 \log n)$ and depth $O(\log^2 n)$ over any absorptive semiring. In other words, the depth decreases significantly for parallelization, while the size of the circuit only blows up only by a logarithmic factor. Note that the circuit depth $O(\log^2 n)$ is optimal due to Theorem~\ref{thm:Karchmer_Widgerson}.

\begin{theorem}\label{sec:rpq:matrix}
The provenance polynomial of any fact of the $\TC$ program over any absorptive semiring $\bS$ can be computed by a circuit of size $O(n^3 \log n)$ and depth $O(\log^2 n)$, where $n$ is the size of the active domain.
\end{theorem}


\begin{proof}
We are essentially proving the absorptive semiring analogue for $\TC \in \operatorname{NC}^2$. As before, let $x_{i,j}$ be the variable that corresponds to the edge $E(i,j)$. Given the input, consider the matrix $M$ given by 
$$M_{i,j} := \begin{cases}
    \mathbf{1},& \text{if } i = j\\
    x_{i,j}, & \text{if } E(i,j) \text{ exists}\\
    \mathbf{0}, & \text{otherwise}\end{cases}$$
Define the matrix multiplication over $\bS$ as the usual matrix multiplication over $\bC$ except replacing $+$ and $*$ by $\oplus$ and $\otimes$. Then, the $(s,t)$-entry of the matrix $M^n$ over $\bS$ computes the provenance polynomial for $\TC$ for the fact $T(s,t)$. This can be seen by inductively showing that the $(s,t)$-entry of $M^{\ell}$ computes the provenance polynomial of all walks of length no greater than $\ell$ from $s$ to $t$. The absorption property is crucial and will be used in two different places: $(i)$ a monomial of a walk from $s$ to $t$ will be absorbed by a monomial  of a path therein (in particular we can safely ignore self-loops), and $(ii)$ the $(i,i)$-entry of $M^{\ell}$ always equals to $\mathbf{1}$.

    It remains to show how to compute $M^n$ over $\bS$ by a circuit for size $O(n^3 \log n)$ and depth $O(\log^2 n)$. This uses the standard trick of parallelization, which we include for completeness. First, any sum of $n$ elements in $\bS$ can be computed in a ``binary tree fashion'' by a circuit using only $O(\log n)$ many $\oplus$-gates with $O(\log n)$ depth. Then, matrix multiplication between two $n$-by-$n$ matrices over $\bS$ can be computed by a circuit using $O(n^3)$ many $\otimes$-gates and $O(n^2 \log n)$ many $\oplus$-gates with depth $O(\log n)$. Indeed, we can compute the $n^3$ many entry-wise product in parallel and then sum the $n^2$ many $n$-sum in parallel, where during each $n$-sum we apply the previous circuit. Finally, to compute $M^n$, we use the method of repeated squaring, where we compute $M^2, M^4, M^8, \dots$. We need to compute $O(\log n)$ many matrix multiplication between two $n$-by-$n$ matrices over $\bS$. Therefore, there exists a circuit using $O(n^3\log n)$ many $\otimes$-gates and $O(n^2 \log^2 n)$ many $\oplus$-gates with depth $O(\log^2 n)$.
\end{proof}

When the language that corresponds to the $\RPQ$ is finite, we know from Theorem~\ref{thm:bounded:depth} that we can construct a circuit of depth $O(\log |I|) = O(\log n)$. We next show that in this case we can construct a circuit of linear size and logarithmic depth. This circuit is asymptotically optimal.\footnote{Since we require internal gates in the circuit have fan-in 2, we need logarithmic depth to read all inputs.}

\begin{theorem}\label{thm:finite_L}
If the regular language $L$ of an $\RPQ$ is finite, then  we can construct a circuit of size $O(m)$ and depth $O(\log n)$ that computes the provenance polynomial of a fact over any semiring $\bS$. 
\end{theorem}

\begin{proof}
Consider the Datalog program $\Pi$ for $L$ in left-linear form. The observation is that for a given fact $T(s,t)$, performing a magic-set rewriting will lead to an equivalent program $\Pi'$ where the IDBs are unary (indeed, after the rewriting $s$ will replace the variable in the leftmost position of any IDB). This program $\Pi'$ has grounding of size only $O(m)$. Hence, to construct a circuit it suffices to consider a constant number of layers, each one of size only $O(m)$. 
\end{proof}

\subsection{Lower Bounds for Regular Path Queries}

Our main result in this section states that any $\RPQ$ whose regular language $L$ is infinite has exactly the same circuit size and depth complexity as $\TC$. In other words, $\TC$ is {\em the} question for $\RPQ$ in terms of circuit complexity. Note that Theorem~\ref{thm:dichotomy_RPQ} implies a dichotomy of circuit depth for $\RPQ$s: a depth-optimal circuit for an $\RPQ$ either has depth $\Theta(\log n)$ or $\Theta(\log^2 n)$, with nothing in between. Furthermore, it is decidable which class of the two an $\RPQ$ belongs to, since it is decidable whether a regular language is finite.

\begin{theorem}\label{thm:dichotomy_RPQ}
Let $\Pi$ be a basic chain Datalog program that corresponds to an infinite regular language $L$. Then, the provenance polynomial for $\Pi$ has the same circuit depth and size complexity as $\TC$ over any absorptive semiring $\bS$.
\end{theorem}

\begin{proof}
    We begin with the reduction from $\TC$ to $\Pi$.
    Since $L$ is an infinite regular language, 
    by the pumping lemma for regular languages~\cite{Sipser97}, there exists an integer $p \geq 1$ such that every string of length at least $p$ can be written as $xyz$ such that:
    \begin{enumerate}
        \item $|y| \geq 1$;
        \item $|xy| \leq p$;
        \item $(\forall n \geq 0) (xy^nz)$ is accepted.
    \end{enumerate}
    Formally speaking, our proof will consist of two steps. 
    First, given an EDB instance $I$ for $\TC$ and a fact $T(s,t)$, we  construct an EDB instance $\bar{I}$ and a fact $\bar{T}(\bar{s}, \bar{t})$ for $\Pi$ such that $T(s,t)$ is $\TRUE$ for $\TC$ on $I$ if and only if $\bar{T}(\bar{s}, \bar{t})$ is $\TRUE$ for $\Pi$ on $\bar{I}$.
    Second, we need to show that reduction implies a depth-preserving and size-preserving reduction between circuits. That is, given the input of $\TC$ as tagged variables and a circuit for $\Pi$, we need to construct a circuit for $\TC$ of no more depth in an asymptotic sense. 

    For the first step we will apply the pumping lemma. Let $I$ be the input for $\TC$.
    We obtain $\bar{I}$ by performing the following transformation, which expands each edge into a longer labeled path:
    \begin{enumerate}
        \item Introduce $|x|$ many new distinct vertices $s_0, s_1, s_2, \dots, s_{|x|-1}$. Label the edges between $s_{i}$ and $s_{i+1}$, where $0 \leq i \leq |x|-1$, by $x_i$ where $s_{|x|} := s$ and $x = x_1x_2\dots x_{|x|}$;
        \item Introduce $|y|-1$ many new vertices $u_1, u_2, \dots, u_{|y|-1}$ between each edge $(u,v)$. Label the edges between $u_i$ and $u_{i+1}$, where $0 \leq i \leq |y|-1$, by $w_i$ where $u_0 := u$, $u_{|y|} = v$ and $y = y_1 y_2 \dots y_{|y|}$. Note that $|y| \geq 1$, therefore the reduction is always valid;
        \item Introduce $|z|$ many new distinct vertices $t_1, t_2, \dots, t_{|z|}$. Label the edges between $t_{i}$ and $t_{i+1}$, where $0 \leq i \leq |z|-1$, by $z_i$ where $t_{0} := t$ and $z = z_1z_2\dots z_{|z|}$.
    \end{enumerate}
    Finally, the fact that corresponds to $T(s,t)$ is $\bar{T}(s_0, t_{|z|})$.
    By construction, the desired equivalence holds. Furthermore, the instance construction can easily encoded as construction of EDBs. It remains to show this reduction can be translated to a size-preserving and depth-preserving reduction between circuits. Note that a general circuit that solves $\Pi$ must solve $\Pi$ when the input graph is restricted to the exact form given by the above reduction. Then, given the input of $\TC$ as tagged variables and a circuit $\bar{F}$ for $\Pi$, we mimic the instance construction to give a circuit $F$ for $\TC$ by doing the following:
    \begin{enumerate}
        \item The input variables, the output gate and the edges between $\oplus$-gates and $\otimes$-gates remain intact;
        \item Whenever an edge in $\bar{F}$ connects to an input variable $x$ for $\Pi$, connect {\em one of the edges} to the input variable $x'$ in $F$ for $\TC$ where $x$ is an edge in the expansion of $x'$ during the transformation, and connect all other edges to $\mathbf{1} \in \bS$.
    \end{enumerate}
    Over $\bB$, the constructed circuit $F$ returns $\TRUE$ if and only if there exists a path from $s$ to $t$. Moreover, over any absorptive semiring $\bS$, $\bar{F}$ computes the provenance polynomial defined as the sum over all path from $s$ to $t$ of the product over all edge variables within each path, due to the structure of the instance for $\Pi$ and the language $L$. Therefore, $F$ correctly computes the provenance polynomial for $T(s,t)$ over $\bS$. We remark the above reduction is not a formula reduction, since the constant input $\mathbf{1}$ will be used $\Theta(m)$ times.

    \smallskip

  We now show the reduction from $\Pi$ to $\TC$, following the same two steps above. Given an input graph $I$ for $\Pi$, we first construct the {\em product graph} 
 between $I$ and the Deterministic Finite Automaton $\DFA$ that recognizes the language $L$. This product graph has size $O(m)$ and $O(n)$ vertices. Let $K \geq 1$ be the number of accept states in the $\DFA$. Then, we run $\TC$ $K$ times over this product graph, once for each accept state $t$, where for a fact $T(u,v)$ the starting node is $(u, s_0)$ with $s_0$ the initial state of the automaton and the end node is $(v,t)$. Then, we take the union over the $K$ results. To make this a circuit and formula reduction that preserves size and depth, we simply do the following:
     \begin{enumerate}
        \item The input variables, the output gate and the edges between $\oplus$-gates and $\otimes$-gates remain intact;
        \item Let an edge in the product graph be of the form $((v_1, v_2), (s_1,s_2))$, where $(v_1, v_2)$ is an edge in $I$ and $s_1$ to $s_2$ is a valid state transition in the $\DFA$. Whenever an edge in the circuit of $\TC$ connects to the input variable for $((v_1, v_2), (s_1,s_2))$, connect that edge to the input variable for $(v_1, v_2)$ in the $\RPQ$. That is, connecting the input variables based on its projections to $G$.
    \end{enumerate}
    Finally, we take the $\oplus$-sum over $K$ such circuits, one for every accept state.  A moment of reflection should convince the reader that the resulting circuit indeed answers $\Pi$ correctly and it is of the same size and depth. Indeed, any path/walk in the product graph corresponds to a valid path/walk in the labeled graph that satisfies $L$.  
\end{proof}

Combining Theorem~\ref{thm:size_depth_tradeoff} and Theorem~\ref{thm:dichotomy_RPQ}, we immediately obtain super-polynomial lower bounds for the formula size.

\begin{theorem}
If the regular language of an $\RPQ$ is infinite, then its equivalent Datalog program $\Pi$ has a $\Omega(n^{\log n})$ formula size lower bound.
\end{theorem}

Note that even though our upper bound is super-polynomial in $n$, there is still some gain since the provenance polynomial might contain exponentially many monomials and thus the trivial way to construct a formula requires truly exponential size. 


\subsection{Bounds for General Basic Chain Datalog}



In this section, we focus on general basic chain Datalog. Our main result is the following circuit depth lower bound for any unbounded basic chain Datalog program.

\begin{theorem}\label{thm:unbounded_CFG}
    If a basic chain Datalog program $\Pi$ is unbounded, then there exists a set of EDBs and a fact $\alpha$  such that the circuit for the provenance polynomial of $\alpha$ over any absorptive semiring $\bS$ has depth $\Omega(\log ^2n)$, where $n$ is the size of the active domain.
\end{theorem}

\begin{proof}
    Roughly, the proof consists of the same idea as in the proof of Theorem~\ref{thm:dichotomy_RPQ}. 
    Recall that the pumping lemma for $\CFG$s~\cite{Sipser97} says that there exists some integer $p \geq 1$ such that every accepted string $s$ that has a length of $p$ or more can be written as $s = uvwxy$ with sub-strings $u,v,w,x,y$, such that: \begin{itemize}
        \item $|vx| \geq 1$;
        \item $|vwx| \leq p$;
        \item $uv^iwx^iy$ is accepted for all integer $i \geq 0$.
    \end{itemize}
    We follow the two-steps approach in the proof of Theorem~\ref{thm:dichotomy_RPQ}. For the first step, we will show an instance-level reduction from $\TC$ to $\CFG$. It will be without loss of generality to assume $|v| \geq 1$, since $|vx| \geq 1$, in the construction that we are going to explain. Also, observe that when $w$ and $x$ are both the empty string, the $\CFG$ degenerates to a regular language, which is addressed by Theorem~\ref{thm:dichotomy_RPQ}. Therefore, we can assume $|wx| \geq 1$.

    Let $I$ be an $(\ell,n)$-layered directed graph as the input to $\TC$. We construct $\bar{I}$ as the input for the $\CFG$ by the following transformation:
    \begin{enumerate}
        \item Let $p:= uv$. Introduce $|p|$ many new distinct vertices $s_0, s_1, s_2, \dots, s_{|p|-1}$. Label the edges between $s_{i}$ and $s_{i+1}$, where $0 \leq i \leq |p|-1$, by $_i$ where $s_{|p|} := s$ and $p = p_1p_2\dots p_{|p|}$;
        \item Introduce $|v|-1$ many new vertices $u_1, u_2, \dots, u_{|v|-1}$ between each edge $(u,u')$. Label the edges between $u_i$ and $u_{i+1}$, where $0 \leq i \leq |y|-1$, by $w_i$ where $u_0 := u$, $u_{|v|} = u'$ and $v = v_1 v_2 \dots v_{|v|}$. Note that $|v| \geq 1$, therefore the reduction is always valid;
        \item Let $q:= wx^\ell y$. Introduce $|q|$ many new distinct vertices $t_1, t_2, \dots, t_{|q|}$. Label the edges between $t_{i}$ and $t_{i+1}$, where $0 \leq i \leq |q|-1$, by $q_i$ where $t_{0} := t$ and $q = q_1q_2\dots q_{|z|}$.
    \end{enumerate}
    Finally, the fact that corresponds to $T(s,t)$ is $\bar{T}(s_0, t_{|q|})$.
    By construction, the desired equivalence holds. Furthermore, the instance construction can easily encoded as construction of EDBs. The reason that this instance-level reduction can be made into a size-preserving and depth-preserving reduction between circuits is similar to the proof of Theorem~\ref{thm:dichotomy_RPQ}: Keep all internal gates and wires between internal gates intact; For each wire that connects to an input variable $x = (u,u')$, connect the wire to the input variable $(u,u_1)$ in transformation step (2); Connect all other wires to the constant $\mathbf{1} \in \bS$. The same argument in the proof of Theorem~\ref{thm:dichotomy_RPQ} shows the constructed circuit computes the provenance polynomial of $\Pi$ over $\bS$ correctly. 
\end{proof}

Combining Theorem~\ref{thm:unbounded_CFG} and Theorem~\ref{thm:size_depth_tradeoff}, we obtain that any unbounded basic chain Datalog program has super-polynomial formula size lower bound over $\bB$. {The following theorem is specified over $\bB$ since a lower bound over $\bB$ implies a lower bound over any absorptive semiring.}

\begin{theorem}
    Any formula for an unbounded basic chain Datalog program has size at least $\Omega(n^{\log n})$ over $\bB$.
\end{theorem}

\section{Upper and Lower Bounds for General Datalog}
\label{sec:upper}

In this section, we prove both upper and lower bounds on circuit depth for Datalog programs that go beyond basic chain Datalog programs.

\subsection{Upper Bounds for Datalog with the Polynomial Fringe Property}

We start by defining the notion of a polynomial fringe as introduced in~\cite{UllmanG88}.

\begin{definition}[Polynomial Fringe Property]
Let $\Pi$ be a Datalog program with target $T$. We say that $\Pi$ has the polynomial fringe property if for every input $I$ of size $m$, the size of any tight proof tree of a fact in $T$ is polynomial in $m$.
\end{definition}

In other words, programs with the polynomial fringe property have "small" tight proof trees. Any linear Datalog program has the polynomial fringe property~\cite{UllmanG88}. It can also be shown that the data complexity of evaluating any program  with the polynomial fringe property is $\textsf{NC}^2$~\cite{UllmanG88}. Hence, such Datalog programs are inherently parallel. Here, we show an analogue of this result for circuits over absorptive semirings.

\begin{theorem}\label{thm:pfp_small_depth_circuit}
Let $\Pi$ be a Datalog program with target $T$ that satisfies the polynomial fringe property. Let $\bS$ be an absorptive semiring. Then, for any input $I$ and any fact of $T$, we can construct a circuit for the provenance polynomial $p^I_\Pi(\alpha)$ over $\bS$ with polynomial size and depth $O(\log^2 |I|)$.
\end{theorem}

\begin{proof}
The idea is to apply the construction of Ullman and Van Gelder~\cite{UllmanG88} for a fast parallel algorithm for Datalog programs with the polynomial fringe property. 

We will first ground the program $\Pi$. Let $\mathcal{N}$ be the set of IDB facts in the grounding. We assign for every IDB fact $\alpha$ a unique id $\id{\alpha}$. We will also introduce a special id $\id{0}$. It will be convenient to represent a grounded rule as $\alpha \obtainedfrom \wedge_i \beta_i  \wedge_j \gamma_j$, where $\beta_i$'s are the IDB facts, and $\gamma_j$'s the EDB facts.

Intuitively, the circuit will keep track of a directed graph $G$ with vertices the ids from $\mathcal{N} \cup \{0\}$. The idea is that the edge $G(\id{0},\id{\alpha})$ will hold the correct value of IDB fact $\alpha$ at the end of the computation. The other edges of the form $G(\id{\beta},\id{\alpha})$ will help to speed up the computation.

The circuit is constructed in $K$ stages (we will define $K$ at a later point). For stage $0$, we define every edge in $G$ to have value $\mathbf{0}$, that is, for every $\alpha,\beta \in \mathcal{N} \cup \{0\}$, $G^{(0)}(\id{\beta}, \id{\alpha}) \gets \mathbf{0}$. For stage $k=1, \dots, K$, we do the following three things. 

First, for any IDB fact $\alpha \in \mathcal{N}$:
$$G_1^{(k)}(\id{0}, \id{\alpha}) \gets \bigoplus_{\alpha \obtainedfrom \wedge_i \beta_i  \wedge_j \gamma_j} \left( \bigotimes_{i} G^{(k-1)}(\id{0}, \id{\beta_i}) \otimes \bigotimes_{i} x_{\gamma_j} \right)$$

Second, for any pair of IDB facts $\alpha,\delta \in \mathcal{N}$:
$$G_1^{(k)}(\id{\delta}, \id{\alpha}) \gets  \bigoplus_{\alpha \obtainedfrom \delta \wedge_i \beta_i  \wedge_j \gamma_j} \left( \bigotimes_{i} G_1^{(k)}(\id{0}, \id{\beta_i}) \otimes \bigotimes_{j} x_{\gamma_j} \right)$$

Third, compute the new graph. For any  $\alpha, \beta  \in \mathcal{N} \cup \{0\}$:
$$G_2^{(k)}(\id{\alpha}, \id{\beta}) \gets 
G^{(k-1)}(\id{\alpha}, \id{\beta}) \oplus G_1^{(k)}(\id{\alpha}, \id{\beta})$$

Fourth, we compute one step of transitive closure on $G$. For any  $\alpha, \beta  \in \mathcal{N} \cup \{0\}$:
$$ G^{(k)}(\id{\alpha}, \id{\beta}) \gets G_2^{(k)}(\id{\alpha}, \id{\beta}) \oplus \bigoplus_{\gamma} \left( G_2^{(k)}(\id{\alpha}, \id{\gamma})\otimes G_2^{(k)}(\id{\gamma}, \id{\beta}) \right) $$
We obtain the value of IDB fact $\alpha$ by reading the gate $G^{(K)}(\id{0}, \id{\alpha})$ at the last stage $K$.

Note that each stage can be implemented via a sub-circuit of polynomial size and $O(\log |I|)$, since each summation needs logarithmic depth to be computed.

We have not defined what is the number of necessary stages $K$. Let $\{T_1, \dots, T_m\}$ be the set of tight derivation trees for an IDB fact $\alpha$. Recall that because of the polynomial fringe property, the size of every tight derivation tree $T_i$ is polynomial in the input $I$. Ullman and Van Gelder~\cite{UllmanG88} show that we need $\log_{4/3} |T_i|$ stages to compute correctly $T_i$. Hence, if we set $K = \max_i \{\log_{4/3} |T_i|\}$, the IDB fact will be computed correctly. Moreover, $K$ must also be $O(\log |I|)$.
\end{proof}


\begin{corollary}
Let $\Pi$ be a linear Datalog program with target $T$. Let $\bS$ be an absorptive semiring. Then, for any input $I$ and any IDB fact of $T$, we can construct a circuit for the provenance polynomial $p^I_\Pi(\alpha)$ over $\bS$ with polynomial size and depth $O(\log^2 |I|)$.
\end{corollary}

Theorem~\ref{thm:pfp_small_depth_circuit} also tells us that the formulas for programs with the polynomial fringe property can be of size $O(m^{\log m})$, thus subexponential in $m$. Not all programs with the polynomial fringe property are linear. 

\begin{example}
One interesting and important example of a non-linear program with the polynomial fringe property is the basic chain Datalog program with grammar expressions with matching parentheses, $S \gets () \mid (S) \mid SS$, also known as {\em Dyck-1 reachability}: 
\begin{align*}
   S(x,y) & \obtainedfrom L(x,z) \wedge R(z,y) \\
   S(x,y) & \obtainedfrom L(x,w) \wedge S(w,z) \wedge R(z,y) \\ 
   S(x,y) & \obtainedfrom S(x,z) \wedge S(z,y) 
\end{align*}
Since this program both has the polynomial fringe property, and its grammar is infinite, we obtain that circuits for Dyck-1 reachability have depth $\Theta(\log^2 m)$.
\end{example}

{We remark that our work does not yield a full dichotomy on whether formulas of polynomial size exist for a fixed program. Indeed, for a bounded program, we have polynomial-size and logarithmic-depth formulas in terms of formula complexity; at the same time, we show in this paper that some unbounded programs admit a super-polynomial lower bound on formula size.}

\subsection{Lower Bounds for Monadic Linear Datalog}


A rule $r$ in a monadic linear Datalog program is of the form $Q_1(x) \obtainedfrom a \wedge Q_2(y)$ where the $Q_i$'s are IDBs and $a$ is the set of EDBs in the rule; $Q_2$ could be potentially absent, in which case the rule is an initialization rule. Given a set of atoms $F$, the {\em variable graph} for $F$ is an undirected graph $G_F = (V,E)$ where $V$ is the set of variables appears in $F$ and two variables are connected by an edge if and only if they both appear in one atom in $F$. We say that a rule $r$ is {\em connected} if $G_A$ is a connected graph and $x, y$ are vertices in $G_A$. A Datalog program is called {\em connected} if every rule is connected. Most Datalog programs seen in practice are connected. We prove the following classification result:

\begin{theorem}\label{thm:monadic_linear_boundedness}
Let $\Pi$ be a monadic linear connected Datalog program, and $\bS$ be a $\otimes$-idempotent absorptive semiring. Then, the target predicate $Q$ of $\Pi$ admits a circuit of depth $\Theta(\log^ 2 m)$ if $\Pi$ is unbounded, and $\Theta(\log m)$ if it is bounded.   
\end{theorem}

An application of Theorem~\ref{thm:pfp_small_depth_circuit} gives us the $O(\log^2 m)$ depth upper bound. To prove  the lower bound, observe that from Theorem~\ref{cor:bounded:equiv} unboundedness over $\bS$ is equivalent to unboundedness over $\bB$. Thus, in the rest of the section we will show an $\Omega(\log^2 m)$ circuit depth lower bound over $\bB$, which will then transfer up to any $\otimes$-idempotent absorptive semiring using Proposition~\ref{prop:positive_semiring}.

\smallskip

Given a linear monadic program $\Pi$ with target $Q$, recursive bodies $\{a_1, \dots, a_m\}$ and initialization bodies $\{b_1, \dots, b_n\}$, we consider the expansions of $\Pi$ into $\CQ$s that sequentially apply some recursive rules and finish with a initialization rule~\cite{CosmadakisGKV88}, in the spirit of Theorem~\ref{thm:bound:idem}. Formally, the {\em recursive alphabet} $\Sigma_{\Pi, r}$ is the set $\{a_1, \ldots, a_m\}$, and {\em initialization alphabet} $\Sigma_{\Pi, i}$ is the set $\{b_1, \ldots, b_n\}$. The {\em alphabet} $\Sigma_{\Pi}$ is the union of $\Sigma_{\Pi, r}$ and $\Sigma_{\Pi, i}$. Then, the set of expansions of $\Pi$ can be viewed as a language over $\Sigma_{\Pi}$; we denote this language as $\operatorname{expand}(\Pi)$. Furthermore, let $C$ be a $\CQ$ denoted by a word of $(\Sigma_{\Pi, r})^*$ or a word of $\left(\Sigma_{\Pi, r}\right)^*\left(\Sigma_{\Pi, i}\right)$. We say that $C$ is {\em accepted} by $\Pi$, if there is some expansion $C_i$ such that there exists a homomorphism from $C_i$ to $C$, which implies the result of $C$ is contained in $C_i$ over $\bB$~\cite{ChandraM77}. Let $\operatorname{accept}(\Pi)$ be the language over alphabet $\Sigma_{\Pi}$ consisting of the words which correspond to $\CQ$s accepted by $\Pi$. The language $\operatorname{notaccept}(\Pi)$ is the complement of $\operatorname{accept}(\Pi)$ over $(\Sigma_{\Pi}^*)$. Note that $\operatorname{expand}(\Pi) \subseteq \operatorname{accept}(\Pi)$. The following result proved in~\cite{CosmadakisGKV88} characterizes the boundedness of a linear connected monadic Datalog.

\begin{theorem}[Cosmadakis, Gaifman, Kanellakis \& Vardi, 88']
    Let $\Pi$ be a linear connected monadic program. Then, $\Pi$ is unbounded over $\bB$ if and only if for every $k>0$, there is a word $w$ in $\operatorname{expand}(\Pi)$ such that $|w|>k$ and the prefix of $w$ of length $k$ is in $\operatorname{notaccept}(\Pi)$.
\end{theorem}
We shall also make use of the following proposition proved in~\cite{CosmadakisGKV88}.

\begin{proposition}\label{prop:expand_notaccept_regular}
    Both the language $\operatorname{expand}(\Pi)$ and $\operatorname{notaccept}(\Pi)$ are regular.
\end{proposition}

We are now ready to prove the lower bound via a circuit reduction from transitive closure. The main idea is to encode each edge between two layers by the canonical database of the expansions instead of by a sequence of edges as in the case of basic chain Datalog.

\begin{theorem}\label{thm:lower_bound_monadic_linear}
Let $\Pi$ be a monadic linear connected Datalog program, and $\bS$ be a $\otimes$-idempotent absorptive semiring. If $\Pi$ is unbounded over $\bB$, then there exists an input $I$ and a fact $t$ in the target $Q$ such that the circuit for the provenance polynomial of $t$ has depth $\Omega(\log ^2 m)$.
\end{theorem}


\begin{proof}[Proof of Theorem~\ref{thm:lower_bound_monadic_linear}]
    We note that a recent result~\cite{AGRZ24} shows $L$-hardness for an extension of monadic linear Datalog with linear temporal logic, using also the tools from~\cite{CosmadakisGKV88} and reducing from undirected reachability.
    For this proof, we reduce from the transitive closure $\TC$ to $\Pi$ where $Q$ is the designated target. Suppose $G$ is the input $(\ell, n)$-layered directed graph and $s, t$ are the distinguished vertices. Since $\operatorname{notaccept}(\Pi)$ is a regular language, its complement $\operatorname{accept}(\Pi)$ is also a regular language. We will assume for now that within each recursive rule $U_1(X)\obtainedfrom q \wedge U_2(Y)$, we have $X \neq Y$; we shall argue for the general case later in this proof.
    
    We first claim that there exists a word of the form $xyzu$ where $x,y,z,u$ are words over $\Sigma_\Pi$ in $\operatorname{accept}(\Pi)$ such that, for any integer $i \geq 0$:
    \begin{enumerate}
        \item the word $xy^izu$ lies in $\operatorname{accept}(\Pi)$;
        \item any non-trivial prefix of $xy^izu$ lies in  $\operatorname{notaccept}(\Pi)$.
    \end{enumerate}
    Indeed, by Theorem~\ref{thm:monadic_linear_boundedness}, we know that for any $k$ there exists a word $w$ in $\operatorname{expand}(\Pi)$ such that $|w| > k$ and the prefix of $w$ of length $k$ is in $\operatorname{notaccept}(\Pi)$. Due to Proposition~\ref{prop:expand_notaccept_regular}, we can choose $k$ large enough such that the prefix of $w$ is of the form $xy^iz$ by invoking the pumping lemma for $\operatorname{notaccept}(\Pi)$. Say $w = xy^izu$ and consider the $\DFA$ $M$ of $\operatorname{notaccept}(\Pi)$. Since $\operatorname{notaccept}(\Pi)$ is the complement of $\operatorname{accept}(\Pi)$, the $\DFA$ that flips in $M$ all accepting state to be non-accepting and all non-accepting state to be accepting, we arrive the $\DFA$ $M'$ for $\operatorname{accept}(\Pi)$. Recall that $\operatorname{expand}(\Pi) \subseteq  \operatorname{accept}(\Pi)$ and consider now the run of $w$ on $M'$. The run must end in an accepting state, and if any word between $xy^iz$ and $xy^izu$ end in an accepting state, we can truncate $w$ to be this word and the new $w$ will again satisfy Theorem~\ref{thm:monadic_linear_boundedness}. This finishes the proof of the claim.


    Consider now the expansion $\CQ$ that corresponds to the word $xyzu$. Because all rules are monadic, linear and connected, 
    this $\CQ$ can be decomposed into three connected $\CQ$s $C_x(X_1, X_2)$, $C_y(X_2, X_3)$, $C_{zu}(X_3,X_4)$ (corresponding to the words $x, y, zu$ respectively) such that no other variables are shared between the $\CQ$s. Notice that the expansion $\CQ$ for the word $xy^2zu$ would then be $C_x(X_1, X_2) \wedge C_y(X_2, X_2') \wedge C_y(X_2', X_3)  \wedge C_{zu}(X_3,X_4)$.  We now construct the input $I$ for $\Pi$, given the layered graph $G$, as follows:
    \begin{enumerate}
        \item For each edge between $s$ and a vertex $v$ in the first layer of $G$, we introduce the canonical database for $C_x(X_1, X_2)$ into $I$, where we identify $X_1$ with $s$ and $X_2$ with the vertex  $v$. Crucially, the remaining values introduced are distinct fresh values.  
        \item For each edge $(v_1, v_2)$ between any two consecutive layers in $G$, we introduce the canonical database for $C_y(X_2,X_3)$ into $I$, where we identify $X_2$ with $v_1$ and $X_3$ with $v_2$.
        \item For each edge between a  vertex $v$ in the top layer and $t$, we introduce the canonical database for $C_{zu}(X_3,X_4)$ into $I$, where we identify $X_3$ with $v$ and $X_4$ with $t$.
    \end{enumerate} 
    We now claim that $\Pi$ over $I$ contains the distinguished domain element $t$ if and only if $s$ and $t$ are reachable in $G$. Indeed, since we fix the distinguished element $t$, $\Pi$ contains $t$ over $I$ if and only if there exists a homomorphism from some expansion of $\Pi$ to $I$. By the construction, this implies there exists a homomorphism from some expansion of $\Pi$ to a $\CQ$ that is the expansion corresponding to a prefix $w'$ of $w$, which means $w' \in \operatorname{accept}(\Pi,Q)$. However, since $\Pi$ is connected, this homomorphism extends to a homomorphism to a non-trivial prefix of $xy^izu$, contradicting to the construction of $w$. This finishes our claim. Finally, the construction of the instance can be made into a circuit reduction as in the proof of Theorem~\ref{thm:dichotomy_RPQ} (in every canonical database we create for an edge $(u,v)$, one fact gets the value of the variable $x_{u,v}$ for $\TC$ and the remaining facts are set to $\mathbf{1}$).

    We now argue for the general case where it is not necessarily $X \neq Y$ for every recursive rule $U_1(X)\obtainedfrom a\wedge  U_2(Y)$. We claim that, upon the existence of a word $xy^iu$ in $\operatorname{accept}(\Pi)$ and any prefix of it in $\operatorname{notaccept}(\Pi)$, such a word also has the property that there exists a rule in the expansion of $y$ satisfying $X \neq Y$. Note that this property guarantees the reduction by encoding of the canonical databases between each layer. Such a word exists because any $\CQ$ corresponding to the expansion $y^i$ of rules in $y$ which all have $X = Y$ will be equivalent. Indeed, they simply repeat the same filtering of the join of some EDBs. Therefore, every rule in all loops in the $\DFA$ for $\operatorname{expand}(\Pi)$ having $X = Y$ contradicts to the assumption that $Q$ is unbounded. This finishes our proof.
\end{proof}

\section{Related Work}
\label{sec:related}

\introparagraph{Provenance for Datalog} The seminal paper of Green et. al~\cite{GreenKT07} first introduced the idea of using provenance semirings for Datalog programs. Recent work has looked into efficient computation of provenance values in Datalog over different 
 semirings~\cite{RamusatMS21,RamusatMS22}, and other work has studied the complexity of why-provenance for Datalog~\cite{CLPS24}.

\introparagraph{Datalog and Semirings} In a fundamental result, Khamis et. al~\cite{KhamisNPSW24} characterized when a $\operatorname{Datalog}^\circ$ program converges in a fintie number of steps depending on it underlying Partially Ordered Pre-Semiring (POPS), a generalization of naturally ordered semirings. This work introduced the notion of a $p$-stable semiring. Interestingly, the $\TC$ program also has a crucial role in the characterization. Recent work has constructed tighter bounds on the convergence~\cite{ImM0P24} when the Datalog program is linear. There has also been progress in the direction of obtaining faster runtime results for Datalog over semirings~\cite{ZhaoDKRT24}.

\introparagraph{Boundedness of Datalog} The boundedness of a Datalog program has been studied intensively in the literature~\cite{CosmadakisGKV88, GaifmanMSV93, HillebrandKMV91, Naughton89, Naughton86, NaughtonS87, Vardi88}. Most of the work focused on finding classes of Datalog programs where deciding boundedness is a tractable property, since in general it is undecidable.

{\introparagraph{Semiring Provenance} The area of capturing provenance via semirings has recently seen significant advancements. For the reader's convenience, we emphasize several key contributions in this field~\cite{GradelT24, DannertGNT21, DannertGNT19}}

\introparagraph{Circuit Complexity Classes}
In a sense, our paper considers the separation of Datalog programs w.r.t. the semiring analogues of the complexity classes $\operatorname{NC}^1$ and $\operatorname{NC}^2$. These classes (and in general $\operatorname{AC}^i$ and $\operatorname{NC}^i$) allow for circuits with Boolean operators $\{\wedge,\vee,\neg\}$. The use of negation is critical, since for monotone Boolean circuits there are linear depth lower bounds~\cite{RazW92} and superpolynomial lower bounds~\cite{razborov1985lower}. However, no such lower bounds are known for Boolean circuits with negation. The circuits in $\operatorname{AC}^i$ and $\operatorname{NC}^i$ are also {\em uniform}, which means that they can be generated efficiently for a given input length -- typically in logarithmic space. In our case, we construct circuits for a specific input (all our constructions in this paper are also efficiently computable).

Related to query evaluation, it is known that $\UCQ$ evaluation is in the class $\operatorname{AC}^0$ that has circuits of constant depth and unbounded fan-in. Ullman and Van Gelder~\cite{UllmanG88} showed that evaluation of linear Datalog -- and any program with the polynomial fringe property -- is in $\operatorname{NC}^2$. 

\section{Conclusion}

In this paper, we studied the question of which Datalog programs have low-depth circuits that represent their provenance polynomials. Several questions remain open.

The most important question is whether we can show that every Datalog program either admits $O(\log m)$-depth circuits or $\Omega(\log^2 m)$-depth circuits. Are there programs with intermediate circuit depth? This also ties to the question of whether it is possible to have $O(\log m)$-depth circuits for unbounded Datalog programs. 

We should also note that we are not aware of any circuit depth lower bounds stronger than the $\Omega(\log^2 m)$ bound for a problem that can be encoded in Datalog. However, we expect that there are Datalog programs that do not have polylogarithmic depth circuits. There exist problems in polynomial time that have stronger lower bounds -- for example, graph matching has monotone Boolean circuits of depth $\Omega(n)$, where $n$ is the number of vertices~\cite{RazW92}. However, graph matching has super-polynomial-size circuits, while all Datalog programs have a polynomial-size circuit. Hence, we would need to find a problem with a poly-size circuit, but depth $\Omega(n^{\epsilon})$ for some $\epsilon>0$.
\section*{Acknowledgments}
This work was done in part while Koutris and Roy were visiting the Simons Institute for the Theory of Computing in Berkeley. This research was partially supported by NSF IIS-2008107, NSF IIS-2147061, and the DeWitt Graduate Fellowship. 


\bibliographystyle{ACM-Reference-Format}
\bibliography{ref}


\end{document}